\documentclass[a4paper, 10pt, conference]{ieeeconf}      % Use this line for a4
\IEEEoverridecommandlockouts                              % This command is only
\usepackage[utf8]{inputenc}
\usepackage[T1]{fontenc}

 \usepackage{graphicx}      % include this line if your document contains figures
 \usepackage{subcaption} % include in preamble
\usepackage{cite}         % <-- Add this line (best for ieeeconf)

\usepackage{url} 
 \usepackage{color}
 \usepackage{changes}

 \usepackage[colorlinks=true, linkcolor=black, urlcolor=black, citecolor=blue]{hyperref}

\usepackage{balance}       % Optional but helpful

\usepackage{amsmath,amssymb,amsfonts,mathrsfs}
\newtheorem{theorem}{Theorem} 

 \newtheorem{assumption}{Assumption}
\newtheorem{definition}{Definition}
 \newtheorem{proposition}{Proposition}

\newtheorem{example}{Example}
\newtheorem{problem}{Problem}
\newtheorem{remark}{Remark}
\usepackage{comment}
\usepackage{bm}
\allowdisplaybreaks[4]
\usepackage{tcolorbox}
\newcommand{\der}{{\rm{d}}}

\title{\LARGE \bf
Towards Lag Consensus with Noisy Digital Twins Perception in Second-order Multi-agent Cyber-physical Systems
}

\author{Zhicheng Zhang$^{1}$,
Fausto Francesco Lizzio$^{2}$,
Zhongjun Ma$^{3}$,
and Masaaki Nagahara$^{4}$% <-this % stops a space
\thanks{*This work was partly supported by Tateisi Science and
Technology Foundation, and National Natural Science
Foundation of China under Grant No.~62466011.}% <-this % stops a space
\thanks{$^{1}$Department of Electrical Engineering, Kyoto University, Katsura, Japan 
(e-mail:~{\tt\small zhang.zhicheng.2c@kyoto-u.ac.jp})}%
\thanks{$^{2}$Department of Mechanical and Aerospace Engineering,
Politecnico di Torino, Italy 
(e-mail:~{\tt\small fausto.lizzio@polito.it})}
 \thanks{$^3$School of Mathematics and Computing Science, Guilin University of Electronic Technology, China 
 (e-mail:~{\tt\small mazhongjun@guet.edu.cn})}
\thanks{$^{4}$Graduate School of Advanced Science and Engineering, Hiroshima University, Japan
(e-mail:~{\tt\small nagahara@ieee.org})}
}

\begin{document}

\maketitle
\thispagestyle{empty}
\pagestyle{empty}

%%%%%%%%%%%%%%%%%%%%%%%%%%%%%%%%%%%%%%%%%%%%%%%%%%%%%%%%%%%%%%%%%%%%%%%%%%%%%%%%
\begin{abstract}
In this paper, we study second-order lag consensus in multi-agent  cyber-physical networks subject to random noise and input failures, within a framework modeling the interactions and perceptions between physical twins and digital twins. We propose a lag consensus protocol and establish sufficient conditions for the mean-square (exponential) stability of the resulting stochastic lag error dynamics. The consensus criteria are derived via Lyapunov analysis using the It\^{o} formula, ensuring robustness to random perturbations and intermittent input failures. Numerical examples illustrate the effectiveness of the proposed method. %consensus strategy.

\end{abstract}

\section{Introduction}
The stability and network modeling of cyber-physical systems (CPS) can be examined by analyzing collective behaviors such as consensus, synchronization, tracking, and formation in multi-agent systems (MAS), e.g., see ~\cite{lewis2013cooperative,bullo2018lectures,nagahara2024control}. Within the networked control systems society, recent attention has focused on the cooperative design, dependability, and resilient control synthesis of physical twins (PT) in the real-world physical layer and digital twins (DT) in the cyber layer within the agent-based CPS framework \cite{pasqualetti2015control,dibaji2019systems,cui2023resilient}.

Over the past decade, lag consensus or delayed tracking problem % \cite{Wang1}, 
is a central topic in the multi-agent cyber-physical networks and has seen significant progress in cooperative control synthesis for deterministic and uncertain dynamical networks. The motivation idea of lag consensus stems from collision avoidance and congestion prevention in leader–following MAS by embedding a time delay in the leader’s dynamics, such that the followers’ states track those of the leader with a specified delay,~see,~e.g.,~\cite{Wang1}. 
Subsequently, some related studies have been extended to cluster lag consensus \cite{ma2016cluster}, fixed-time lag consensus \cite{ni2017fixed}, lag-bipartite consensus \cite{Bhowm}, and 
successive lag synchronization \cite{li2023successive}.
%component-lag consensus \cite{wang2023p}. 
In addition, various lag consensus based control laws have been proposed, including adaptive pinning control \cite{wang2016pinning}, robust event-driven control \cite{yu2025lag}, sliding mode control \cite{fu2025time}, and ${H}_{\infty}$ control~\cite{wang2025lagTCNS}. 
Obviously, robust control methods for lag consensus has drawn growing attention, as uncertainty is an inevitable factor of practical systems. However, there are few results on modeling it using stochastic dynamics, which require more resilient control paradigms and some kind of stability criteria. An important issue is how to incorporate stochastic noise into network modeling in a lag consensus perspective, while still allowing the problem to be solved using standard and general ways, such as those based on stochastic differential equation (SDE) stability theory \cite{Oksendal2003}.

In this paper, we introduce a new viewpoint on the integration of multi-agent lag consensus and cyber-physical networks, formulated within a framework involving PT-DT transmissions and interactions. These two paradigms can be jointly utilized to analyze and design lag consensus control protocols. The physical layer consisting of PTs corresponds to real-world agent-based networks, such as connected and autonomous vehicles (CAVs), unmanned aerial vehicle (UAV) swarms, intelligent robotic systems, multiple satellite systems, and internet of things for smart cities, etc \cite{cui2023resilient}. However, these agent-based dynamical networks are inherently coupled with their DT perceptions, where state estimation, sensing, and communication \cite{vatanski2009networked} are inevitably subject to the uncertainty, random noise, or stochastic perturbations \cite{wang2023survey}. In a nutshell, the PT-DT interaction in the CPS can be characterized by a dynamic network model with a similar or same network topology, where the true states in the physical layer (PTs) are perceived through the estimated states in the cyber layer (DTs). In fact, such perceptions can be affected by stochastic perturbations in the noisy coupling layer, arising from sensing jitter, estimation errors, and uncertainties (see~Fig.~\ref{fig:DT-PT-MACPS}). 

% \textcolor{red}{I would provide a brief introduction of the intermittent input failures before talking about the main contribution. }

In view of this, we aim toward reaching lag consensus with noisy digital-twin perceptions in agent-based cyber-physical systems, where the resulting system is modeled by second-order stochastic dynamics. The novel contributions of this paper are summarized as follows: (i) we propose a resilient stochastic lag consensus protocol in the case of intermittent input failures; (ii) we obtain the mean-square stability criteria for SDE based second-order multi-agent CPS framework in terms of PT-DT interactions; (iii) we prove lag consensus robustness against some classes of nonlinear dynamics in physical layer (PT) simulations.

%=========================Key=======================
% {\bf\color{red}
%  (1) Why introduce CPS PT-DT
%  (2) Why explore lag consensus comapred to other consensus and its developemnts? 
%  (3) Why consider random noise, how to model it? 
%  (4) why investigate intermittent input failures? }

%{\bf\emph{Structure of this paper}}--
This paper is organized as follows. Section~\ref{sec:preliminary} reviews preliminaries and lag consensus concepts. Section~\ref{sec:stochastic-lag-CPS-PT-DT} presents the problem setup for PT-DT network modeling in multi-agent CPS and presents the resilient lag consensus protocol. Section~\ref{sec:main-results-lag-consensus} provides sufficient conditions for stochastic lag consensus based on the stability of the resulting error dynamics governed by SDE. Section~\ref{sec:numerical-examples} gives the numerical examples to illustrate the effectiveness of proposed methods. Finally, we conclude the paper in Section~\ref{sec:conclusion}.

%===============================
% This paper is organized as follows: Section~\ref{sec:preliminary} reviews some preliminaries and lag consensus concepts. Section~\ref{sec:stochastic-lag-CPS-PT-DT} describes the problem setup about network modeling of PT-DT interactions in multi-agent CPS, and presents the resilient lag consensus protocol. In Section~\ref{sec:main-results-lag-consensus}, we provides some sufficient condition of stochastic lag consensus derived from the stability of error dynamics governed SDE. Section~\ref{sec:numerical-examples} we gives the numerical examples to illustrate the main results. Finally, we conclude this paper.

%and give the Appendix

%=================Notations
%\section{Notation}
\emph{Notation:} 
%Let %$\mathbb{N}=\{ 0,1,2, \cdots \}$ and 
%$\mathcal{I}_N=\{ 1,2, \cdots ,N\}$. 
 Let ${I_N}$ (${0_N}$) be the $N$-order identity (zero) matrix and $\mathbf{1}_{n} = (1, \ldots ,1)^{\top} \in \mathbb{R}^{n}$. $\|\cdot\|$ is the Euclidean norm, $\mathbb{E}$ is the expectation and $\otimes$ denotes the Kronecker product. 
% $\mathscr{C}\left([-\tau, 0]; {\mathbb{R}^d}\right)$ is the Banach space $\mathscr{L}(\cdot)$ of all continuous vector-valued functions mapping the intervals $[-\tau, 0]$ into ${\mathbb{R}^d}$ with the norm ${\|\phi\|_{\infty}=\sup _{-\tau \leq \zeta \leq 0}\|\phi(\zeta)\|}$.
For $A \in {\mathbb{R}^{N \times N}}$, ${\lambda _{\max }}(A)$ or ${\lambda _{\min }}(A)$ means the maximum or minimum eigenvalue of matrix $A$ and $A^{\top}$ indicate its transpose. Let ${A^s} = \frac{A + A^{\top}}{2}$ be the symmetrical part of $A$, and $A<0$ ($A\leq0$) means that matrix $A$ is real symmetric negative definite (semi-definite). The trace of matrix $A$ is denoted by $\textbf{tr}(A)$. We denote by $\mathrm{diag}\{{\kappa_1},\ldots ,{\kappa_N}\}$ the diagonal matrix.

%================
% \emph{Graph Theory}:
% Let $\mathcal{G}=(\mathcal{V},\mathcal{E},{A})$ be a weighted digraph of order $N$, with the set of nodes $\mathcal{V}=\{\mathfrak{v}_{1},\ldots,\mathfrak{v}_{N}\}$, the set of edges $\mathcal{E}\subseteq\mathcal{V}\times\mathcal{V}$, and the weighted adjacency matrix ${A}={({a_{ij}})_{N\times{N}}}$ of $\mathcal{G}$. A directed edge of graph $\mathcal{G }$ is denoted by ${e}_{ij}=(\mathfrak{v}_{i},\mathfrak{v}_{j})$ which means that node $\mathfrak{v}_i$ can receive information from node $\mathfrak{v}_j$. We define ${a_{ji}}>0$ if and only if there is a directed edge $\mathfrak{e}_{ij}\in\mathcal{E}$; otherwise, ${a_{ji}}=0$. A digraph contains a spanning tree if there is a node that can reach all the other nodes following the edge directions in graph $\mathcal{G}$. 
% %A digraph $\mathcal{G}$ is called strongly connected if for any two distinct nodes of the graph, there exists a connected directed path. 
% The set of the neighbors of the agent $i$ is denoted by ${\mathcal{N}_i}=\{\mathfrak{v}_j\in\mathcal{V}:(\mathfrak{v}_i,\mathfrak{v}_j)\in\mathcal{E}\}$, e.g., see \cite{bullo2018lectures}.

\section{Preliminaries and Lag Consensus}\label{sec:preliminary}
\subsection{Graph Theory}
%\emph{Graph Theory}:
Let $\mathcal{G}=(\mathcal{V},\mathcal{E},{A})$ be a weighted digraph of order $N$, with the set of nodes $\mathcal{V}=\{\mathfrak{v}_{1},\ldots,\mathfrak{v}_{N}\}$, the set of edges $\mathcal{E}\subseteq\mathcal{V}\times\mathcal{V}$, and the weighted adjacency matrix ${A}={({a_{ij}})_{N\times{N}}}$ of $\mathcal{G}$. A directed edge of graph $\mathcal{G }$ is denoted by ${e}_{ij}=(\mathfrak{v}_{i},\mathfrak{v}_{j})$ which means that node $\mathfrak{v}_i$ can receive information from node $\mathfrak{v}_j$. We define {${a_{ij}}>0$} if and only if there is a directed edge $\mathfrak{e}_{ij}\in\mathcal{E}$; otherwise, {${a_{ij}}=0$}. The set of the neighbors of the agent $i$ is denoted by ${\mathcal{N}_i}=\{\mathfrak{v}_j\in \mathcal{V}:(\mathfrak{v}_i,\mathfrak{v}_j)\in\mathcal{E}\}$.
%e.g., see \cite{bullo2018lectures}.
%A digraph $\mathcal{G}$ is called strongly connected if for any two distinct nodes of the graph, there exists a connected directed path. 

The network communication is described by the adjacency matrix $A=(a_{ij})\in\mathbb{R}^{N\times{N}}$, and the associated graph Laplacian $L=(l_{ij})$ is given by 
$l_{ij}=-a_{ij}$ for $i\ne{j}$ and $l_{ii}=\sum_{j=1,j\neq{i}}^{N}a_{ij}$ for $i=j$. 
A digraph contains a spanning tree if there is a node that can reach all the other nodes following the edge directions.
Introduce a \emph{leader} node $\mathfrak{v}_0$, and an augmented graph $\bar{\mathcal{G}} = (\bar{\mathcal{V}}, \bar{\mathcal{E}})$ such that $\bar{\mathcal{V}} = \{\mathfrak{v}_0, \mathfrak{v}_1, \ldots, \mathfrak{v}_N\}$ 
and $\bar{\mathcal{E}} \subseteq \bar{\mathcal{V}} \times \bar{\mathcal{V}}$. 

%===================Assumption (Network Toplogy)
\begin{assumption}[Network Topology]
\label{assum:network-topology}
The graph $\bar{\mathcal{G}}$ contains a \emph{directed spanning tree} with the leader $\mathfrak{v}_0$ as the root node for all $t \ge 0$.
\end{assumption}
%The spanning tree of the direct netwok see \cite{lews}

%===============
\subsection{Second-order Multi-Agent Networks}
%{\emph{Multi-agent dynamics}}: 
Consider a class of $N$ second-order agents, called \emph{followers}, whose dynamics are governed by 
\begin{align}\label{eq:follower-dynamics}
\begin{cases}
\dot{x}_i(t) \!=\! v_i(t), \\
\dot{v}_i(t) \!=\!f(x_i(t), v_i(t))
+ c_1 \sum\limits_{j\in \mathcal{N}_i} a_{ij}\big (x_j(t) \!-\! x_i(t)\big)\\ 
{\kern 30pt}+\, c_2 \sum\limits_{j\in \mathcal{N}_i} a_{ij}\big (v_j(t) - v_i(t)\big)
+ u_i(t),
\end{cases}
\end{align}
where $x_i(t), v_i(t), u_i(t) \in \mathbb{R}^n$ indicate the position, velocity, and control input of (follower) agent~$i$, $\forall{i}\in\mathcal{V}$, respectively.
The function $f:\mathbb{R}^n \times \mathbb{R}^n \to \mathbb{R}^n$ denotes the intrinsic dynamics of homogeneous agents.

%%-------------Assumption--(Nonlinear)-
\begin{assumption}[Lipschitz Condition]
\label{assum:Lipschitz-cond-nonlinear}
The function $f:\mathbb{R}^{n} \times \mathbb{R}^{n} \to \mathbb{R}^{n}$ is nonlinear and Lipschitz continuous in $(x,v)$; that is, there exist constants $\rho_1, \rho_2 \ge 0$ such that
\begin{align*}
\left\|f(x_1, v_1) - f(x_2, v_2)\right\|
\le \rho_1 \left\|x_1 - x_2\right\| + \rho_2 \|v_1 - v_2\|,
\end{align*}
for all $ x_1, x_2, v_1, v_2 \in \mathbb{R}^{n}.$
\end{assumption}

% \begin{align*}
%    l_{ij}=-a_{ij}~\text{for}~i\ne{j}, \quad
%    l_{ij}=\sum_{j=1,j\neq{i}}^{N}a_{ij} ~\text{for}~i={j}
% \end{align*}
%=======

%%==============leader dynamics==========
The \emph{leader} acts as a virtual dynamic reference whose dynamics is modeled as
\begin{align}
%\begin{cases}
\dot{x}_0(t) = v_0(t),\quad
\dot{v}_0(t) = f(x_0(t), v_0(t)),
%\end{cases}
\label{eq:leader-reference}
\end{align}
where $x_0(t),v_0(t) \in \mathbb{R}^{n}$ denote the position and velocity of the leader, respectively.

Through Assumption \ref{assum:network-topology}, the followers in \eqref{eq:follower-dynamics} track the evolution of the leader described by \eqref{eq:leader-reference}.
%For example, a leader can be defined as a virtual reference by taking the average behavior of all followers, i.e., $x_0(t)=\frac{1}{N}\sum_{j=1}^{N}x_i(t)$, $v_0(t)=\frac{1}{N}\sum_{j=1}^{N}v_i(t)$. 
% In practice, the leader’s trajectory may evolve toward a periodic orbit, an asymptotically fixed point, or a regular (or even a chaotic) attractor. 
Generally, systems \eqref{eq:follower-dynamics} and \eqref{eq:leader-reference} are referred to as \emph{leader–follower} multi-agent systems, also known as \emph{tracking} plant.

\subsection{Review of Classical Lag Consensus Problem}

The main interest of this paper is the so-called \emph{second-order lag consensus} problem, in which the followers’ states track those of the leader with a time delay $\tau \geq 0$, also known as ``lag'' in this article \cite{wang2016pinning}. Setting $t_{\tau}:=t - \tau$, denote $\tilde{x}_i(t) := x_i(t) - x_0(t_\tau)$ and $\tilde{v}_i(t) := v_i(t) - v_0(t_\tau)$ as the lag errors with respect to the position and velocity, respectively. The classical lag consensus problem, without random noise or input failures, aims to achieve
%$x_i(t) - x_0(t_\tau)$ (resp., $v_i(t) - v_0(t_\tau)$), 
\begin{align}
\lim_{t\to\!+\!\infty}\|\tilde{x}_i(t)|=0, ~
\lim_{t\to\!+\!\infty}\|\tilde{v}_i(t)\|=0.
\label{def:lag-consensus-wo-noise}
\end{align}
%The key difference from leader–follower consensus lies in embedding the time delay directly into the leader dynamics in \eqref{eq:leader-reference}. 
This design prompts followers to predict the trajectory so that can avoid collisions and congestion.
Alternatively, we define the compact vector 
\begin{align*}
\xi(t)=(\xi_{1}(t)^{\top},\ldots,\xi_N(t)^{\top})^{\top} \in\mathbb{R}^{2Nn},
\end{align*}
where the $i$-th element $\xi_i(t)=(\tilde{x}_{i}(t)^{\top},\tilde{v}_{i}(t)^{\top})^{\top}\in\mathbb{R}^{2n}$,
then the lag error \eqref{def:lag-consensus-wo-noise} reduces to $\lim_{t\to\infty}\|\xi_i(t)\|\!=\!0$, $\forall{i}\!\in\!\mathcal{V}$.

%=========================lag consensus Meets Random Noise and Input Failures in CPS DT-PT Networks
\section{Lag Consensus under Random Noise and Input Failures in CPS PT–DT Networks}\label{sec:stochastic-lag-CPS-PT-DT}
In this section, we explore second-order lag consensus in multi-agent systems (MASs) subject to stochastic noise and intermittent input failures, within the framework of cyber-physical systems (CPS). In practice, each physical agent or \emph{physical twin} (PT), is associated with a corresponding \emph{digital twin} (DT) that resides in the cyber layer and maintains an estimated or revised version of the PT’s dynamic state \cite{cui2023resilient,liu2024digital}. In this article, we consider a simple scenario by assuming that \emph{each agent in PT is associated one-to-one with an agent in DT, forming a bijective mapping between the physical and cyber layers of the CPS and keeping the same network topology}.
%\deleted{As shown in Fig.~\ref{fig:DT-PT-MACPS}, the PT evolve according to second-order dynamics, while their states are perceived in the DT of cyber layer through stochastic couplings $\sigma_{ij}w(t)$.}

%================CPS-PT-DT-Concept========
\begin{figure}
    \centering
    \includegraphics[width=0.7\linewidth,height=5.2cm]{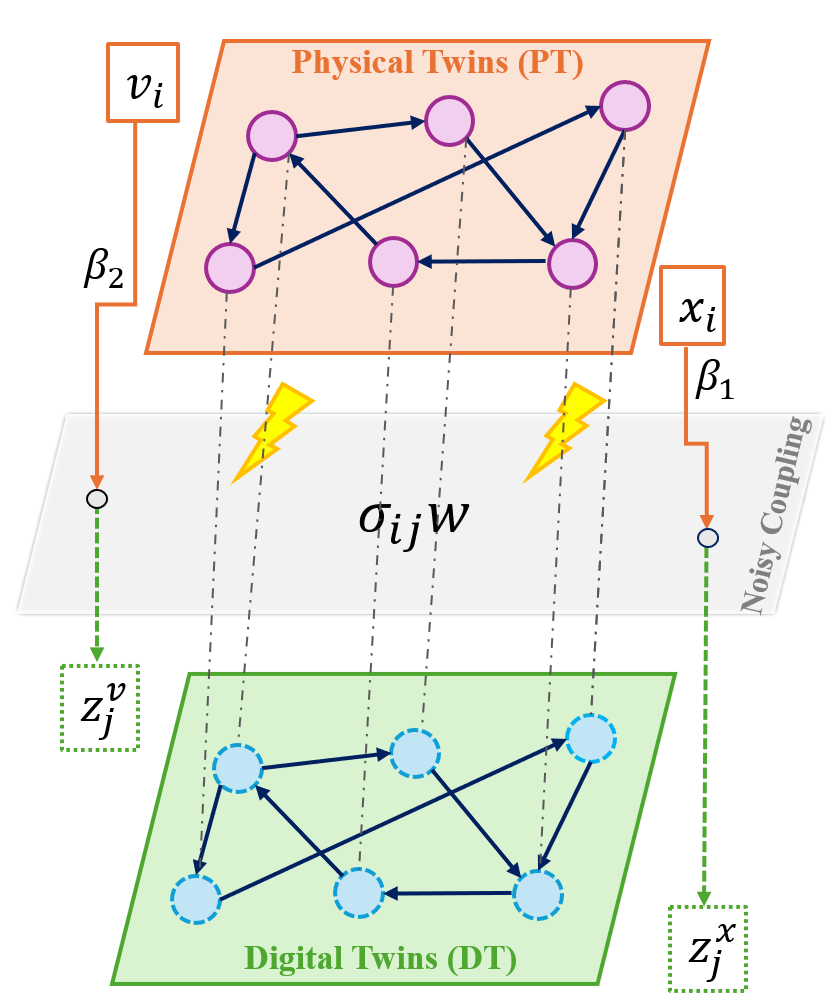}
    \caption{A sketch of PT-DT interaction in second-order agent based CPS under stochastic noise.}
    \label{fig:DT-PT-MACPS}
\end{figure}
%(e.g., small world network)
%=========================

\subsection{PT-DT Interaction in CPS: Stochastic Broadcasting}
%===========================================================
During the PT-DT interaction shown in Fig.~\ref{fig:DT-PT-MACPS}, the PT of agent $j$ represents (or evolves with) the true state $(x_j(t), v_j(t))$ in the physical layer.  
However, due to sensing jitter, disturbances, and estimation errors, the DT of agent $i$ does not access the true neighbor states directly. 
Instead, it perceives stochastic broadcasts of the physical states, modeled as
\begin{align}
\text{(DT):}\qquad 
\begin{aligned}
    z_j^x(t) &= x_j(t) + \beta_1 \sigma_{ij} (x_j(t) - x_i(t))\, w(t), \\
    z_j^v(t) &= v_j(t) + \beta_2 \sigma_{ij} (v_j(t) - v_i(t))\, w(t),
\end{aligned}
\label{eq:stochastic-nois-perb-neigbors}
\end{align}
where $z_j^x(t)$ and $z_j^v(t)$ denote the DT perceived estimates of the physical states $x_j(t)$ and $v_j(t)$ from agent $i$. Here, $w(t) \in \mathbb{R}$ denotes a white noise satisfying $\int_0^t w(s)\, ds = W(t)$, i.e., $\dot{W}(t)=w(t)$, and $W(t)$ is a standard Wiener process or Brownian motion defined on a complete probability space $(\Omega, \mathcal{F}, \mathcal{P})$ with the natural filtration $\{\mathcal{F}_t\}_{t\ge0}$. The parameters $\beta_1 \in \mathbb{R}$ and $\beta_2\in \mathbb{R}$ are the noise intensities for position and velocity, respectively. Finally, the gains $\sigma_{ij}>0$ if and only if $a_{ij}>0$, and $\sigma_{ij}=0$ otherwise. This ensures that noise is only applied along existing communication links.
Similar stochastic perturbation formulations can be found in \cite{Sun,Mu,Oksendal2003}.

% \deleted{Let $D=(d_{ij})$ denote the noise-coupling matrix with $d_{ij}=-\sigma_{ij}~(i\neq j)$, $d_{ii}={+}\sum_{j=1,j\neq{i}}^{N}\sigma_{ij}~(i=j)$.}

Then, a follower agent $i$ employs the noisy neighbor states \eqref{eq:stochastic-nois-perb-neigbors} in its dynamics \eqref{eq:follower-dynamics} instead of the actual states $x_j(t)$ and $v_j(t)$. Thus, the stochastic noise enters through the relative differences between the received DT and actual PT neighbor states in position and velocity, i.e.,
\begin{align*}
    c_1 \sum_{j \in \mathcal{N}_i} a_{ij} \big(z_j^x(t) - x_i(t)\big), \quad
    c_2 \sum_{j \in \mathcal{N}_i} a_{ij} \big(z_j^v(t) - v_i(t)\big).
\end{align*}

\subsection{PT-DT Interaction in CPS: Input Failures}
Setting an initial time $t_0 = 0$, let $[t_k, s_k]$ denote the \emph{control-active} period with duration $s_k - t_k$, while $(s_k, t_{k+1})$ is the \emph{inactive} period of length $t_{k+1} - s_k$. Define the maximal ratio of time with control input failures
\begin{align*}
\psi = \mathop{\lim}\limits_{k \to +\infty} \mathop{\sup}_{k \in \mathbb{N}}
\frac{t_{k+1} - s_k}{t_{k+1} - t_k},
\end{align*}
where $\psi \in (0,1)$, and
%denoting the maximal fraction of time with control failures; 
$\psi = 0$ reduces to continuous control. 
%Here, we focus on $\psi \in (0,1)$.

%---------------------------------------------
% Assumption: Input Failures
%---------------------------------------------
\begin{assumption}\label{assum:input-failures-intermittent}
For the intermittent input failures (possibly not periodic), 
there exist two scalars $\theta$ and $\delta$ such that 
\begin{align}
\inf_{k \in \mathbb{N}} (s_k - t_k) = \theta, \quad
\sup_{k \in \mathbb{N}} (t_{k+1} - t_k) = \delta,
\end{align}
where $0 < \theta < \delta < +\infty$. 
This assumption ensures that the duration of each failed interval 
is bounded by $\delta - \theta$; 
hence, both the control-active and failed intervals are finite. 
Under this assumption, it follows that $\psi \le 1 - {\theta}/{\delta}$.
\end{assumption}

\subsection{PT-DT Interaction in CPS: Problem Definition}

\begin{definition}[Mean-square lag consensus]\label{def:lag-consensusconsensus}
For a predefined lag $\tau \ge 0$, the multi-agent cyber–physical network 
\eqref{eq:follower-dynamics}--\eqref{eq:leader-reference} under PT-DT interaction is said to achieve second-order mean-square lag consensus if, for all $i\in\mathcal{V}$ and any initial states\footnote{Although the error is defined as $\tilde{x}_i(t) = x_i(t) - x_0(t_\tau)$ (resp. $\tilde{v}_i(t)$), its initial value depends on the leader’s past at $t=-\tau$. Formally, the leader’s initial functions over $[-\tau,0]$ may be taken in the Banach space $\mathscr{C}([-\tau,0],\mathbb{R}^{2n})$, but the followers’ error dynamics remain finite-dimensional, and their initials are vectors. For simulations, random vectors for both leader and followers suffice.}
\begin{align}
\lim_{t\to\infty} \mathbb{E}\|\tilde{x}_i(t)\| = 0,
\qquad
\lim_{t\to\infty} \mathbb{E}\|\tilde{v}_i(t)\| = 0.
\label{eq-def:stochastic-lag-consensus-noise}
\end{align}
\end{definition}
%%%%%%%%%%%%%%%%%%%%%%%%%%%%%%%%%%%%%%%%%%%%%%%%
However, stochastic lag consensus \eqref{eq-def:stochastic-lag-consensus-noise} is harder than the standard case \eqref{def:lag-consensus-wo-noise} due to stochastic perturbations and intermittent input failures, especially with large noise or weak connectivity. This motivates the following Problem~\ref{prob:stochastic-lag-consensus}.

%======================Problem present============
%\begin{tcolorbox}
\begin{problem}\label{prob:stochastic-lag-consensus}
Given the second-order MASs~\eqref{eq:follower-dynamics}–\eqref{eq:leader-reference} 
satisfying Assumptions~\ref{assum:network-topology}--%\ref{assum:Lipschitz-cond-nonlinear}--
\ref{assum:input-failures-intermittent}, 
and subject to PT-DT interactions with stochastic perturbations \eqref{eq:stochastic-nois-perb-neigbors} in the CPS framework, 
design a lag consensus protocol with control $u_i(t)$ to hedge against random noise 
and intermittent input failures, and thus achieve mean-square second-order lag consensus \eqref{eq-def:stochastic-lag-consensus-noise}.
\end{problem}
%\end{tcolorbox}

%============================Table============
% \begin{table*}[h!]
% \centering
% \caption{Roles and noise effects in PT--DT structured CPS}
% \renewcommand{\arraystretch}{1.1}
% \begin{tabular}{lccc}
% \hline
% \textbf{Layer} & \textbf{Role} & \textbf{Noise source} & \textbf{Effect} \\
% \hline
% PT  & Real agent dynamics & Noise-free & Receives $u_i(t)$ \\
% DT   & Sensing/communication & Noise $w(t)$ & Generates $z_j^{x},z_j^{v}$ \\
% Control law & Computed in DT, applied to PT & From $z_j^{x},z_j^{v}$ & Induces stochastic response \\
% \hline
% \end{tabular}
% \label{tab:dt-pt-summary}
% \end{table*}

\subsection{Resilient Control Synthesis: Lag Consensus Protocol}
We now introduce our proposed stochastic lag consensus for Problem~\ref{prob:stochastic-lag-consensus}, which is designed by a piecewise, resilient intermittent (on/off) control policy as follows
\begin{align}
&u_i(t)\notag\\
&\!=\!\begin{cases}
\,-\kappa_i\Big(c_1\big(x_i(t)-x_{0}(t_\tau )\big)+c_{2}
\big(v_{i}(t)-v_{0}(t_\tau) \big)\Big)\\
\!-\!\sum\limits_{j\in \mathcal{N}_i} \sigma_{ij}\Big({\beta_1}\big(x_{i}(t)\!-\!x_{j}(t)\big)+\beta_{2}\big(v_{i}(t)\!-\!{v_j}(t)\big)\Big)w(t),\\
{\kern 165pt} t \in [t_{k},s_{k}],\\
0,{\kern 145pt} t \in (s_{k},t_{k+1})
\end{cases}
\label{eq:lag-tracking-controller-PD}
\end{align}
where $\kappa_i > 0$ if agent $i$ receives information from the leader, and 0 otherwise. This corresponds to a pinning-control scheme in which a subset of $\mathcal{V}$ with $\ell$ leader-informed agents is selected~\cite{Porfiri,song2010second,liu1,lewis2013cooperative}.
The first term drives the system toward lag consensus in \eqref{def:lag-consensus-wo-noise}, also known as relative displacement feedback, while in this case it involves a delayed input. The second term perceives the stochastic coupling induced by noisy PT-DT exchanges \eqref{eq:stochastic-nois-perb-neigbors}.

% \deleted{That is, in the cyber-layer (DT), followers receive the leader’s delayed information and stochastic neighbor broadcasts intermittently. These hybrid modes model intermittent input failures typical in CPS-based MASs.}

% %====(CAV interpretation _REAL World PT ===
% \begin{remark}[Practical CAV interpretation]
% In a CAV setting, the PT denotes the actual vehicle dynamics $(x_i,v_i)$, while the DT denotes its cyber estimation used for coordination. The control strategy \eqref{eq:lag-tracking-controller-PD} enables each follower to track the leader’s lagged state, 
% capturing the cyber–physical desynchronization between PT and DT layers. Also, the stochastic coupling term models random communication perturbations among vehicles, and the intermittent on–off modes reflects practical time-triggered or control updates or input failures in vehicular networks.
% \end{remark}

%%------------------Main Results=============
\section{Main results: Stability Analysis}\label{sec:main-results-lag-consensus}
\subsection{Stochastic Modeling of Lag Error Dynamics}
In general, studying the convergence of leader follower trajectories is challenging. Accordingly, we focus on the resulting error dynamics for Problem~\ref{prob:stochastic-lag-consensus}, which describes the lag consensus problem via the next SDE formulation.

\begin{proposition}[Recasting Lag Error Dynamics via SDE]\label{proposit:sde_lag_consensus_error_dynamics}
Consider a class of MASs based CPS described by \eqref{eq:follower-dynamics}--\eqref{eq:leader-reference}, subject to stochastic neighbor broadcasting \eqref{eq:stochastic-nois-perb-neigbors} and a constant delay $\tau \ge 0$. Let the CPS employs the lag consensus protocol with intermittent control \eqref{eq:lag-tracking-controller-PD}.
Then, the lag error dynamics can be recast in a piecewise SDE form, referred to as the stochastic error dynamics:
\begin{align}
\begin{aligned}
\begin{cases}
\der\xi(t)\!=\!\big(\widetilde{F}_\tau(t) + (H_1 \!\otimes\! I_n)\xi(t)\big)\der{t}
            \!-\! (U \!\otimes\! I_n)\xi(t)\,\der{W(t)}, \\
            \hfill t \in [t_k, s_k],\\
\der\xi(t)\!=\!\big(\widetilde{F}_\tau(t) + (H_2 \!\otimes\! I_n)\xi(t)\big)\der{t},
\qquad\quad t \in (s_k, t_{k+1}),
\end{cases}
\label{eq:stochastic-error-dynamics}
\end{aligned}
\end{align}
where $\tilde{L}= L + K$ with $K=\mathrm{diag}\{\kappa_1,\ldots,\kappa_{\ell},0,\cdots,0\}$, and
\begin{align*}
    H_{1}=
    \begin{pmatrix}
        %\begin{smallmatrix}
            0_{N} & I_{N}\\
            -c_{1}\tilde{L} & -c_{2}\tilde{L}
       % \end{smallmatrix}
    \end{pmatrix},\quad
       H_{2}=
    \begin{pmatrix}
      %  \begin{smallmatrix}
            0_{N} & I_{N}\\
            -c_{1}L & -c_{2}L
      %  \end{smallmatrix}
    \end{pmatrix},
\end{align*}
\begin{align*}
    &F(x,v)=(f({x_1},{v_1})^{\top},\ldots,f({x_N},{v_N})^{\top})^{\top},\\
    &\bm{f}_{N}(x_0^{\tau},v_0^{\tau})=\mathbf{1}_N\otimes f((x_0(t_\tau),v_0(t_\tau)),\\
&\widetilde{F}_\tau(t) =
\begin{pmatrix} \mathbf{0}_{nN} \\[2pt]
F(x(t),v(t))-\bm{f}_{N}(x_0^{\tau},v_0^{\tau}))
\end{pmatrix}.
\end{align*}
Also, the noise matrix $D = (d_{ij}) \in \mathbb{R}^{N \times N}$ is defined as 
%$d_{ij}=-\sigma_{ij}$ for $i\neq{j}$, and $d_{ii}=\sum\nolimits_{j = 1}^N \sigma_{ij}$ for $i=j$.
\begin{align}
d_{ij} = 
\begin{cases} 
-\sigma_{ij}, & i \ne j,\\
\sum_{j=1}^N \sigma_{ij}, & i = j,
\end{cases}
~\Longrightarrow~
U= \begin{pmatrix}
            {0_N}&{0_N}\\
            {\beta_1}D &{\beta_2}D
    \end{pmatrix}.
    \label{eq:noise-matrix-D-U}
\end{align}
Thus, realizing lag consensus in Problem~\ref{prob:stochastic-lag-consensus} equals ensuring the mean-square stability of the error dynamics in \eqref{eq:stochastic-error-dynamics}.
\end{proposition}
%=======================
\begin{proof} 
We omit the proof of Proposition~\ref{proposit:sde_lag_consensus_error_dynamics}, as it follows by taking the derivative of the Definition \ref{def:lag-consensusconsensus} together with the well-posed Problem~\ref{prob:stochastic-lag-consensus} and controller in \eqref{eq:lag-tracking-controller-PD}.
\end{proof}
%=======================

Obviously, when the agent-based CPSs in PT-DT experience random noise and input failures, the error system \eqref{eq:stochastic-error-dynamics} evolves as an SDE with both \emph{drift} and \emph{diffusion} terms if the controller is active; otherwise, it reduces to an ordinary differential equation (ODE) containing only the drift term. 
% Therefore, solving Problem~\ref{prob:stochastic-lag-consensus} is equivalent to ensuring the mean-squre strability of error dynamics

% In other words, the MAS network achieves mean-square (or almost-sure) lag consensus, which is equivalent to the trivial solution ${\xi}(t;\phi)$ of the above SDDE \eqref{eq:stochastic-error-dynamics} being mean-square (or almost-sure) stable.

%================PesudoSDDE Remarks=============
\begin{remark}
For the error dynamics \eqref{eq:stochastic-error-dynamics}, the delayed leader term $ f({x_0}(t_\tau),{v_0}(t_\tau))$ is exogenous inputs (does not depend on $\xi(t)$). In fact, the error behavior is defined relative to a delayed reference (lag), but its \emph{error dynamics depends only on the current followers' states rather than its past historical data trajectory}. As a result, the error system is essentially an SDE with an exogenous input delay, i.e.,~``{PseudoSDDE}''.%rather than a true SDDE. 
 Consequently, the equilibrium of the error dynamics is shifted by the delay $\tau$, yet the stability of the system is delay-independent.
\end{remark}

%%=============Main Results (Stability)========
\subsection{Lyapunov based Mean-Square Stability Analysis}%\label{sec:main-results-lag-consensus}
As result in Proposition~\ref{proposit:sde_lag_consensus_error_dynamics}, the stochastic lag consensus of second-order agent CPS \eqref{eq:follower-dynamics}--\eqref{eq:leader-reference} under noisy~\eqref{eq:stochastic-nois-perb-neigbors} can be recast as the mean-square stability of the error dynamics \eqref{eq:stochastic-error-dynamics}, which implies that the controller \eqref{eq:lag-tracking-controller-PD} is effective for Problem~\ref{prob:stochastic-lag-consensus} if it can ensure the stability and robustness of the lag error dynamics. Our main result is as follows.

%-=--=========Theroem===========
% \begin{theorem}\label{thm3}
% Suppose that Assumptions \ref{assum:Lipschitz-cond-nonlinear}--\ref{assum:input-failures-intermittent} hold. The second-order multi-agent systems \eqref{eq:follower-dynamics} and \eqref{eq:leader-reference} with random noise and input failures under intermittent controller \eqref{eq:lag-tracking-controller-PD} can achieve mean square lag consensus if
% \begin{enumerate}
%     \item[i)]  $\gamma\geq0$,
%     \item[ii)]  $\bar{\mu}_{1}\theta-\bar{\mu}_{2}(\delta-\theta)>0$
% \end{enumerate}
% % \begin{align*}
% % &(\mathrm{i}){\kern 4pt}{{2{\lambda _{\min }}(R_1)> \|{U^T}PU\|}},\\
% % &(\mathrm{ii}){\kern 4pt}{{{\bar \mu }_1}\theta- {{\bar \mu }_2}({\delta}-{\theta})}>0,
% % \end{align*}
% where $\gamma := 2 \lambda_{\min}(R_1) - \lambda_{\max}(U^\top P U)$, $\bar{\mu}_{1}=\gamma/\lambda_{\max}(P_2)$, $\bar{\mu}_2=2\lambda_{\max}(R_2)/\lambda_{\min}(P_1)$ with the appropriate matrices ${P}$, ${P_1}$, ${P_2}$, ${R_1}$ and ${R_2}$ are well-defined in (15) 
% \end{theorem}
%========================Theorem =
\begin{theorem}\label{thm:stochastic-lag-consensus-CPS}
Based on error dynamics \eqref{eq:stochastic-error-dynamics},
the lag consensus protocol $u_i(t)$ designed in \eqref{eq:lag-tracking-controller-PD} solves Problem~\ref{prob:stochastic-lag-consensus}, 
that is, ensures mean-square stability and achieves second-order lag consensus of CPS, if the next conditions are satisfied:
\begin{enumerate}
    \item[(i)]  $\gamma := 2 \lambda_{\min}(R_1) - \lambda_{\max}(U^\top P U) \ge 0$,
    \item[(ii)] $\bar{\mu}_{1}\theta - \bar{\mu}_{2}(\delta-\theta) > 0$,
\end{enumerate}
where $\bar{\mu}_{1} = \gamma / \lambda_{\max}(P_2)$, $\bar{\mu}_2 = 2 \lambda_{\max}(R_2)/\lambda_{\min}(P_1)$, 
and the parameters and matrices $P$, $P_1$, $P_2$, $R_1$, and $R_2$ are well-defined in \eqref{eq:Lyapunov-P-matrix}, \eqref{eq:V_fun_lower-upper-bound} and \eqref{eqs:parameters-matrices}.
\end{theorem}

\begin{proof}
    See Appendix~\ref{appendix-proof-thm}.
\end{proof}
\begin{remark}
Theorem \ref{thm:stochastic-lag-consensus-CPS} provides sufficient conditions, where condition (i) for tolerable stochastic noise bound in DT perceptions and condition (ii) for intermittent input failures, ensuring that the multi-agent CPS achieves mean-square lag consensus despite random noise and intermittent control transmission between the leader and followers.
\end{remark}

%===================LMI for Theorem=======
% \begin{remark}[LMI Form]
% To ensure mean-square stability of the error dynamics in Theorem~\ref{thm:stochastic-lag-consensus-CPS}, one can equivalently formulate the conditions using linear matrix inequalities (LMIs). Specifically, find symmetric positive definite matrices $P_1, P_2 > 0$ such that
% \[
% R_1 \otimes I_n - \frac{1}{2} U^\top P U >0, \qquad
% P_1, P_2 > 0,
% \]
% where $R_1$, $U$, and $P$ are defined in \eqref{eqs:parameters-matrices} and \eqref{eq:noise-matrix-D-U}.  

% Satisfying this LMI guarantees that 
% \(\gamma = 2 \lambda_{\min}(R_1) - \lambda_{\max}(U^\top P U) \ge 0\), 
% so that the lag-consensus protocol \eqref{eq:lag-tracking-controller-PD} achieves mean-square second-order lag consensus. This formulation allows systematic computation of controller gains $c_1, c_2, \kappa_i$ and verification of robustness under stochastic perturbations.
% \end{remark}

\section{Numerical Simulations}\label{sec:numerical-examples}
In this section, we illustrate the proposed second-order lag consensus protocol for the multi-agent CPS in numerical examples. We conduct the simulations in Matlab using {Runge-Kutta~4} with the stochastic setup~\cite{higham2001algorithmic}.
%consider several numerical examples to illustrate the proposed second-order lag consensus protocol problem of the multi-agent cyber-physical systems, that is, the stability analysis of error dynamics in terms of SDE. 
\begin{example}[Chua's circuit]
 Consider the nonlinear dynamics $f(x_i,v_i)$ that is modeled as Chua's circuit, which can be thought of as an analog of a micro-electric circuit in smart grids (or as a simple model for studying dynamics in power networks) in real-world physical layers (PT). Thus, the dynamics in \eqref{eq:follower-dynamics} is given by
\begin{align*}
f({x_i},{v_i}) = 
{\small
\left( \begin{array}{l}
9\big(v_{i2}- \frac{{2}}{7}v_{i1}-\frac{{3}}{{14}}(|v_{i1} + 1| - |v_{i1} + 1|) \big),\\
v_{i1} - v_{i2} + v_{i3},\\
-\frac{100}{7}v_{i2},
\end{array} \right)}
\end{align*}
where 
${x_{i}(t)}, {v_{i}(t)}\in \mathbb{R}^{3}$ are the position and velocity states of agents $i\!=\!0,1,\cdots,9$. That is, the network consists of agent $i = 0$ serving as the leader and the root node, while the remaining $N = 9$ agents are follower agents. In this example, we take initial states as random vectors $x_i(0),v_i(0)\in\texttt{rand}(0,1)$, and set the parameters $c_1=c_2=0.2$. Here the Lipschitz constants can be computed as $\rho_1=0$, $\rho_2=\|{\partial{f}}/{\partial{v}}\|=14.28$. Besides, the Laplacian matrix $L$ of network topology is defined as follows
\begin{align*}
L=\begin{pmatrix}
%\begin{smallmatrix}
   0&0&0&0&0&0&0&0&0\\
0&0&0&0&0&0&0&0&0\\
{ - 3}&0&3&0&0&0&0&0&0\\
{ - 1}&{ - 3}&0&6&{ - 2}&0&0&0&0\\
0&{ - 4}&0&0&4&0&0&0&0\\
0&0&{ - 3}&0&0&3&0&0&0\\
0&0&0&0&0&0&1&{ - 1}&0\\
0&0&0&{ - 2}&0&0&0&2&0\\
0&0&0&0&{ - 3}&0&0&0&3
%\end{smallmatrix}
\end{pmatrix},
\end{align*}

In the above network, it is clear that the agents indexed by $\{1,2\}$ do not receive information from other follower nodes. By pinning control these first two agents with the leader, the associated augmented graph $\bar{\mathcal{G}}$ forms a directed spanning tree with the leader as the root node. Hence, we try to adopt $K=\text{diag}(6,6,0,\ldots,0)$ with $\kappa_1=\kappa_2=6$. Also, we take  noise densities $\beta_1=\beta_2=0.2$ that gives the noise coupling matrix $D=L$ and $U$ in \eqref{eq:noise-matrix-D-U}.

% then we can see that those agent indexed by $\{1,2,3,8,9\}$ is can be viewed as control or pinning node candadtions, then the augmented graph $\bar{\mathcal{G}}$ has a directed spanning tree, in which the leader is as the root node. Hence

% Therefore, we adopt $K=\text{diag}(5,5,5,0,\ldots,0,5,5)$. 

% such that by only 
% Regarding the above netwok, it is easily see that the agents indexed by $\{1,2\}$ are isolated agents, and if we set the leader agent to pin or control such two agents, then the related argument graph hasing a diret sapnning tree. 

% controlling the first two agents making the leader agent being a root agent, the argument graph hasing a driet spanning tree. Hence, we here adopt $K=\text{diag}(5,5,0,\ldots,0)$, and noise coupling matrix $D=L$, and the noise densities $\beta_1=\beta_2=0.2$.
% since the dynamics without position information  we can directly, tne noise peratuabution not affect on  position as show nin the 

\begin{figure*}[!htbp]
    \centering
    \begin{subfigure}[b]{0.45\linewidth}
        \centering
        \includegraphics[width=\linewidth]{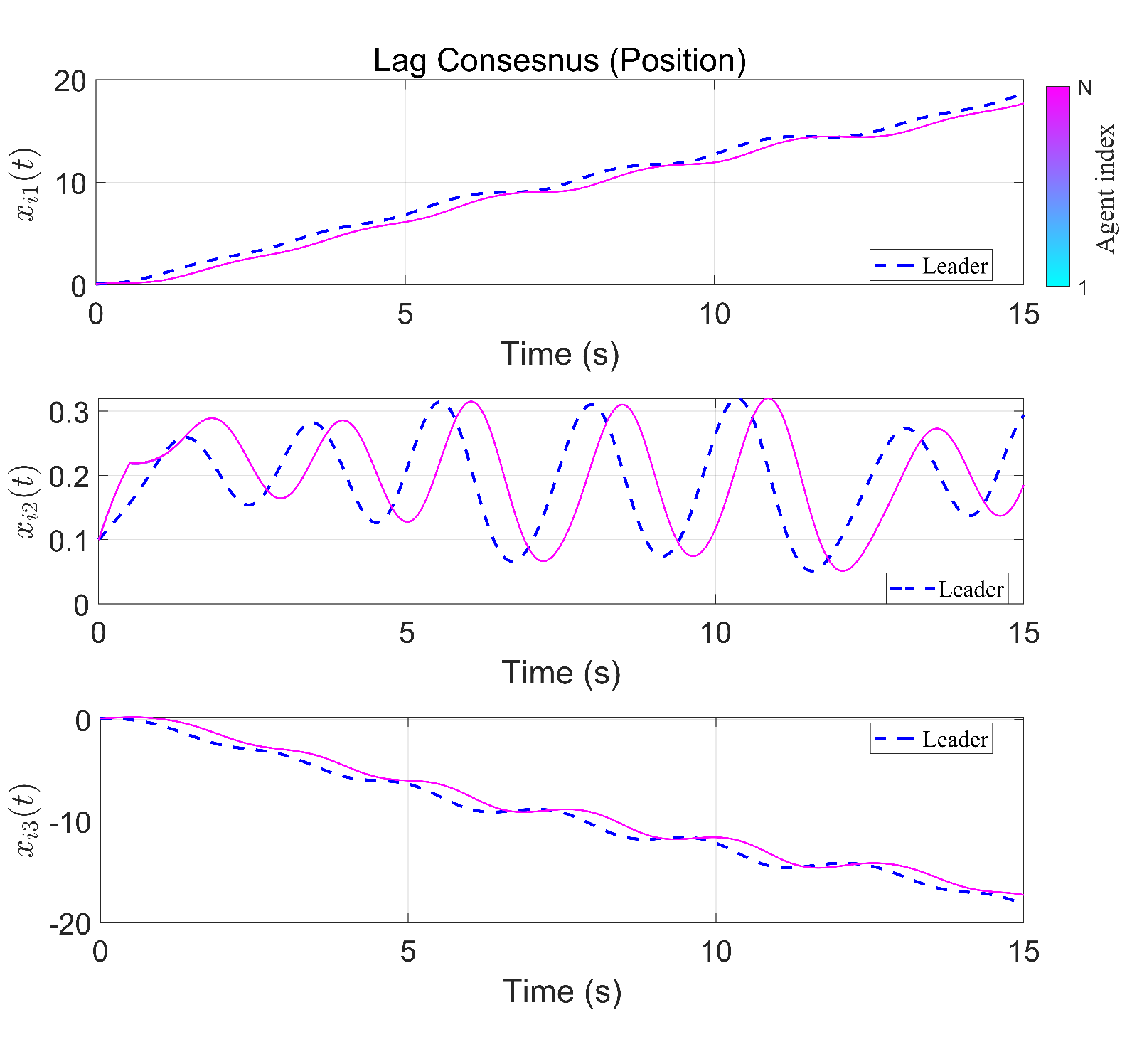}
        \caption{Lag consensus in position}
        \label{fig:lag-position}
    \end{subfigure}
    \begin{subfigure}[b]{0.45\linewidth}
        \centering
        \includegraphics[width=\linewidth]{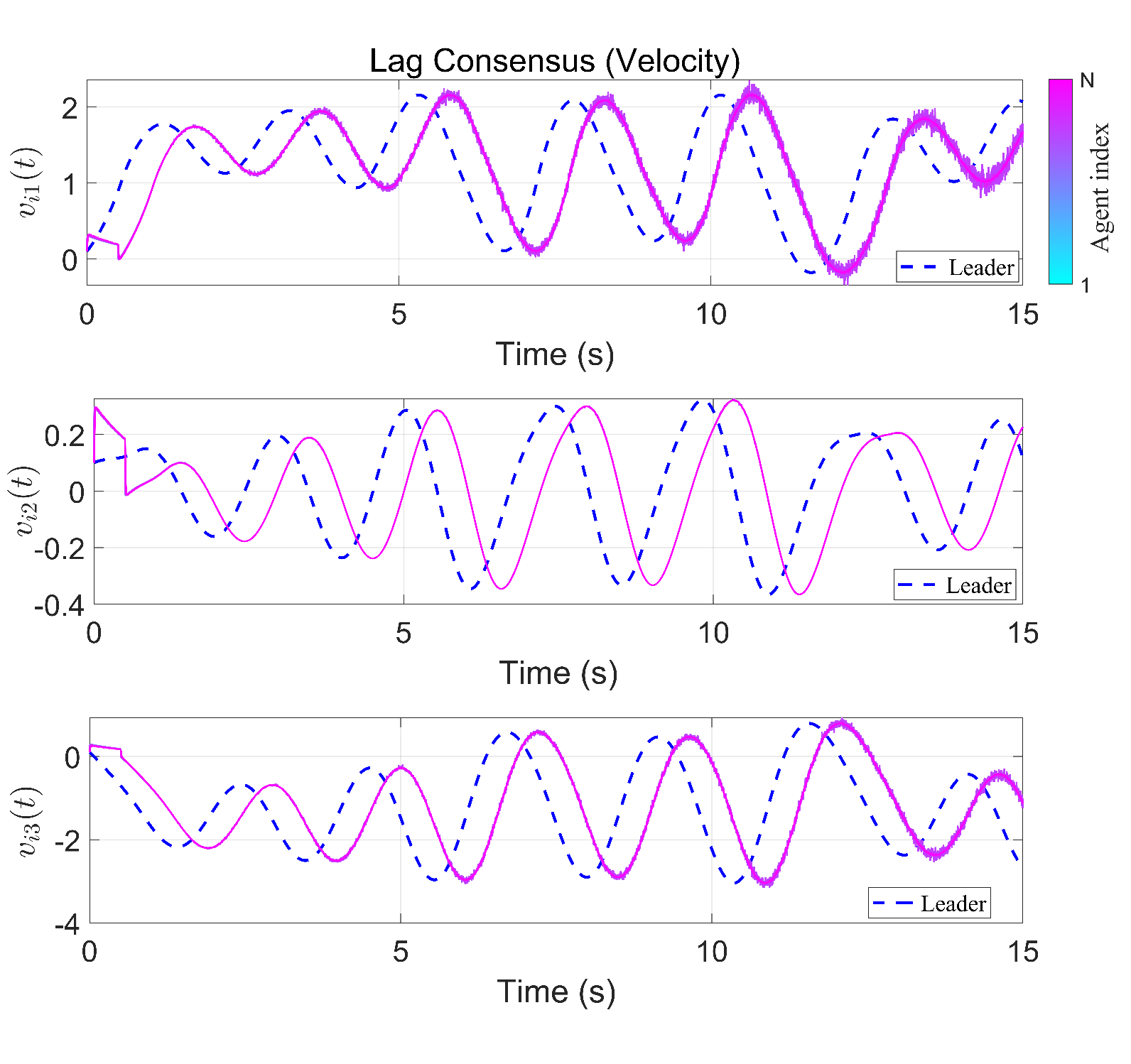}
        \caption{Lag consensus in velocity}
        \label{fig:lag-velocity}
    \end{subfigure}
    \begin{subfigure}[b]{0.45\linewidth}
        \centering
        \includegraphics[width=\linewidth]{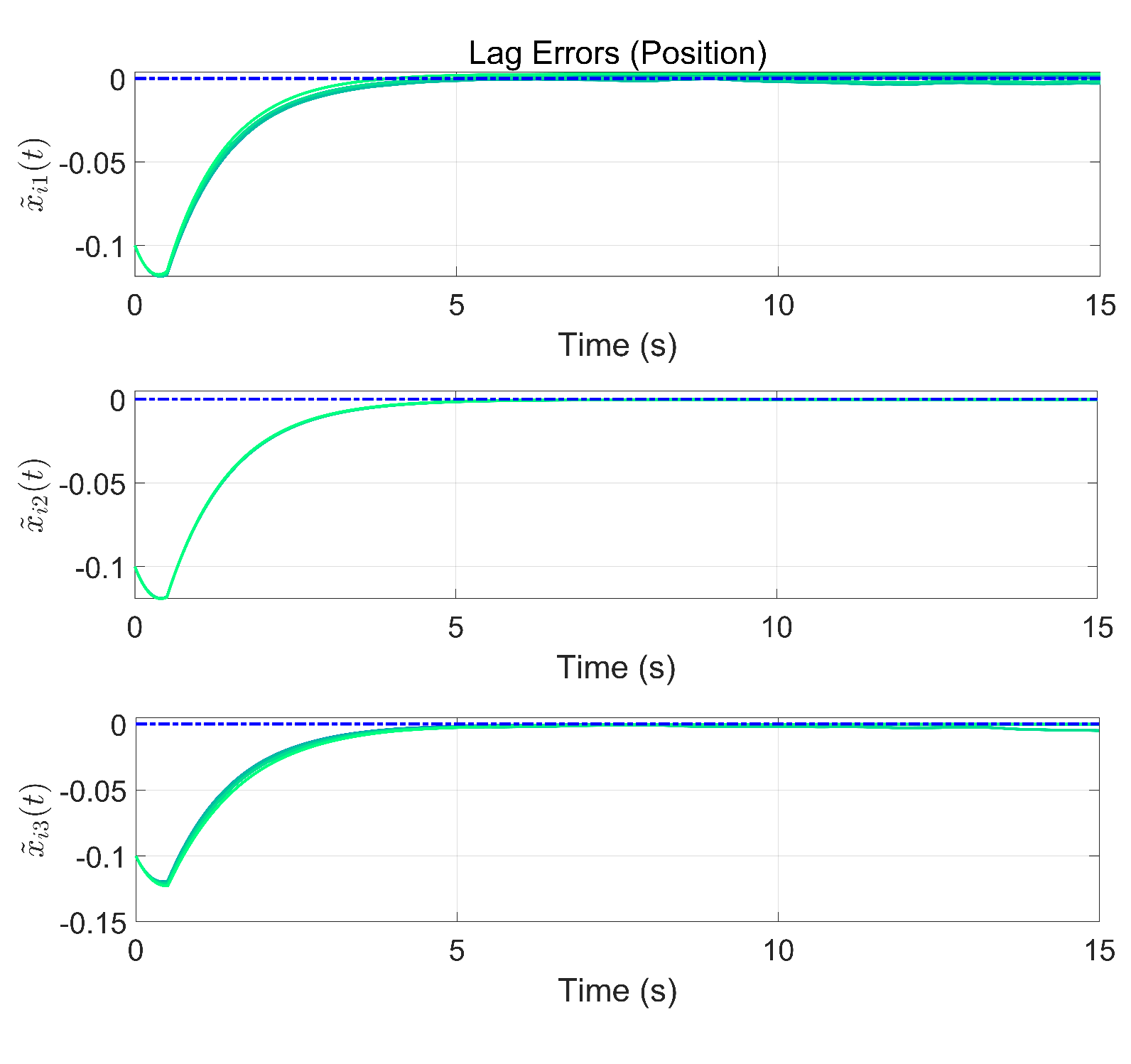}
        \caption{Lag errors in position}
        \label{fig:lag-pos-error}
    \end{subfigure}
    \begin{subfigure}[b]{0.45\linewidth}
        \centering
        \includegraphics[width=\linewidth]{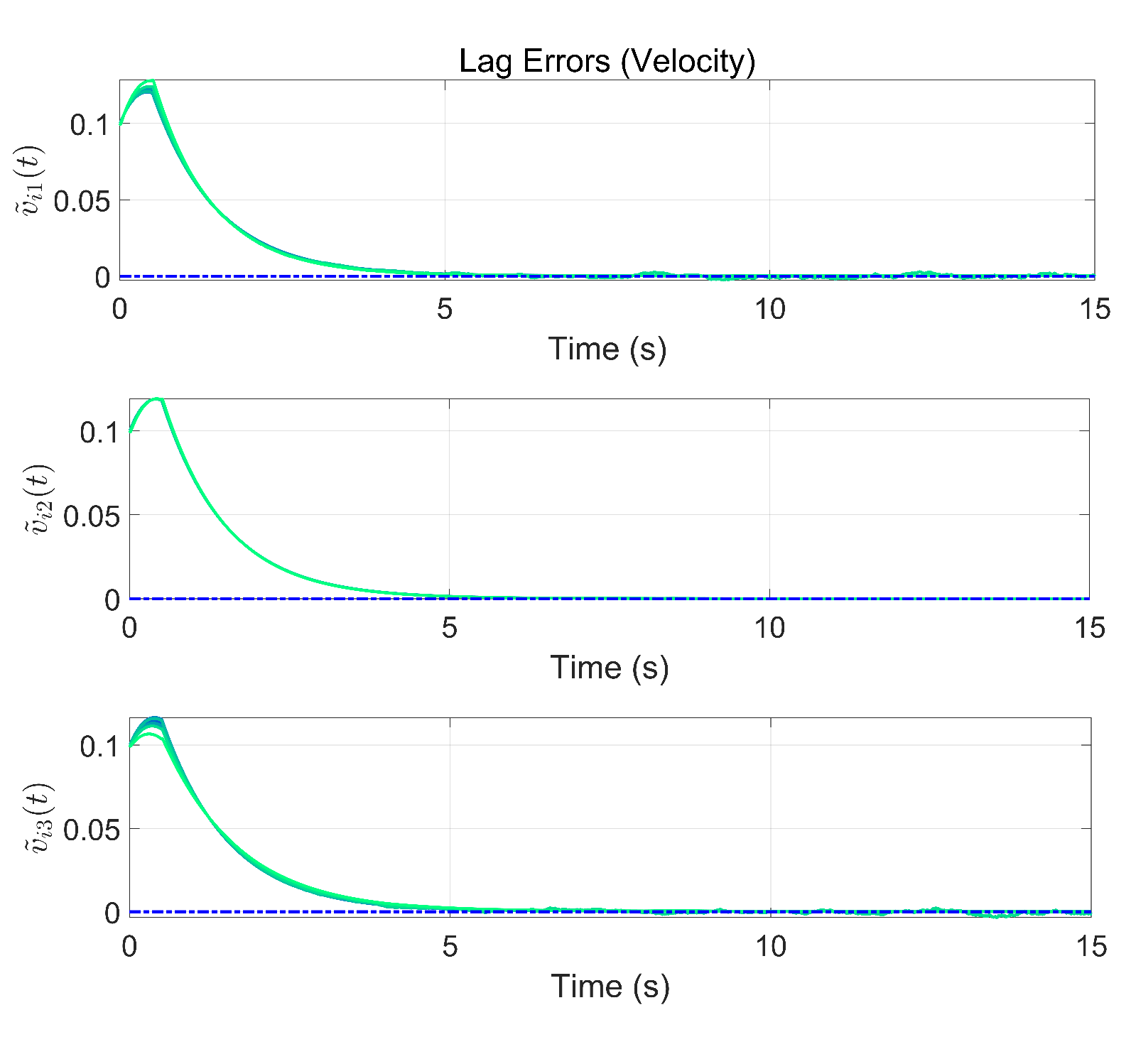}
        \caption{Lag errors in velocity}
        \label{fig:lag-vel-error}
    \end{subfigure}
    \caption{Lag consensus with delay $\tau=1$ {s}: (a) positions, (b) velocities, (c) position lag errors, and (d) velocity lag errors.}
    \label{fig:lag-and-lag-errors}
\end{figure*}

%Besides, we set the intermittent input failures
% =============================================
% Intermittent Control with \psi == 0.7
% =============================================
We simulate the system on the time horizon $t \in [0, 15] \ {s}$, and the intermittent input failures are randomly generated according to Assumption~\ref{assum:input-failures-intermittent}, in which the failure ratio is bounded by $\psi = 0.7$, 
%meaning that the control input may fail for at most 70\% of the total time.
Set $\theta = 0.3$ and $\delta = 1.0$. For each cycle $k = 0, 1,\dots$: $T_k\sim\mathcal{U}[0.3, 1.0]$, $T^{c}_k\sim\mathcal{U}[0.3, T_k]$, $t_{k+1} = t_k + T_k, s_k = t_k + T^{c}_k,$
with $t_0 = 0$ {s}. 

According to Theorem~\ref{thm:stochastic-lag-consensus-CPS}, we adjust the related parameters to meet conditions (i) and (ii). 
With the controller \eqref{eq:stochastic-nois-perb-neigbors} implemented as discussed above, Fig.~\ref{fig:lag-position} and Fig.~\ref{fig:lag-velocity} show the lag consensus simulation, where nine follower agents track the leader’s trajectory with the prescribed delay $\tau=1$ s in position and velocity, respectively. Because the nonlinear dynamics do not use position information (that is, $\rho_1=0$), the stochastic perturbations do not influence the positions, while the velocities are visibly affected by the random noise.
Meanwhile, Fig.~\ref{fig:lag-pos-error} and Fig.~\ref{fig:lag-vel-error} show that the position and velocity mean square lag errors converge exponentially to zero, implying the lag consensus enjoys the robustness.
\end{example}

%============ Small World =============
\begin{figure*}[!htbp]
    \centering
    
    \begin{subfigure}[b]{0.93\textwidth}
        \centering
        \includegraphics[width=\linewidth]{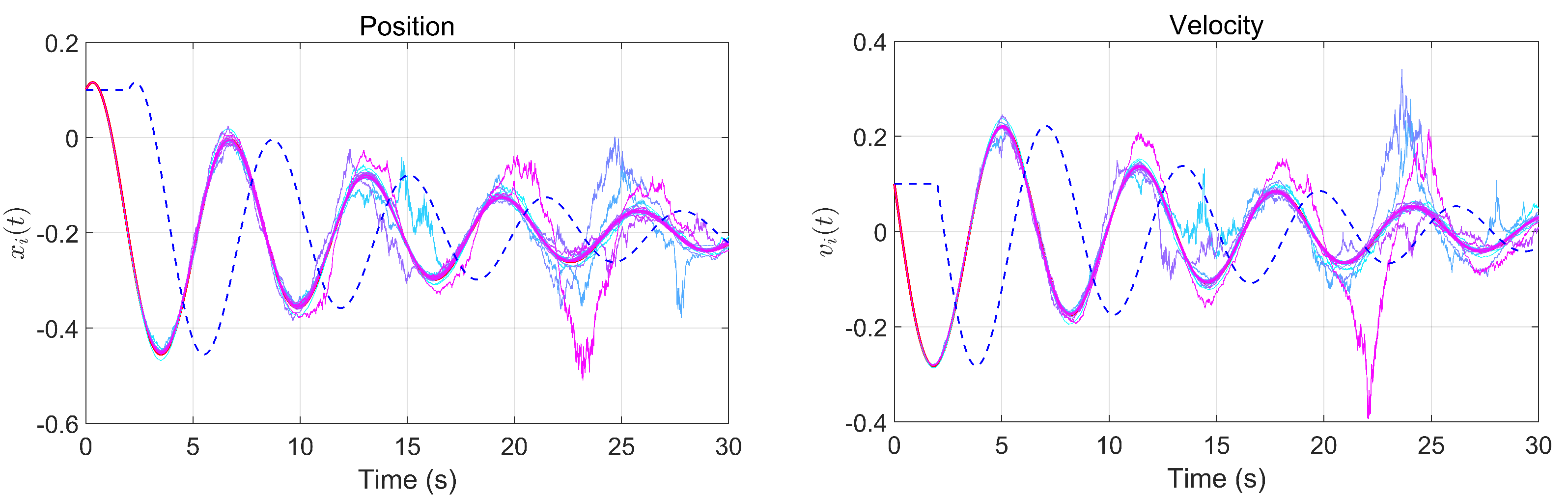}
        \caption{Lag consensus with small feedback gain $\kappa_1=1.7$.}
        \label{fig:lag-con-small-word-xv-3-1}
    \end{subfigure}
    
   % \vspace{8pt}   % adjust vertical space between subfigures
    
    \begin{subfigure}[b]{0.93\textwidth}
        \centering
        \includegraphics[width=\linewidth]{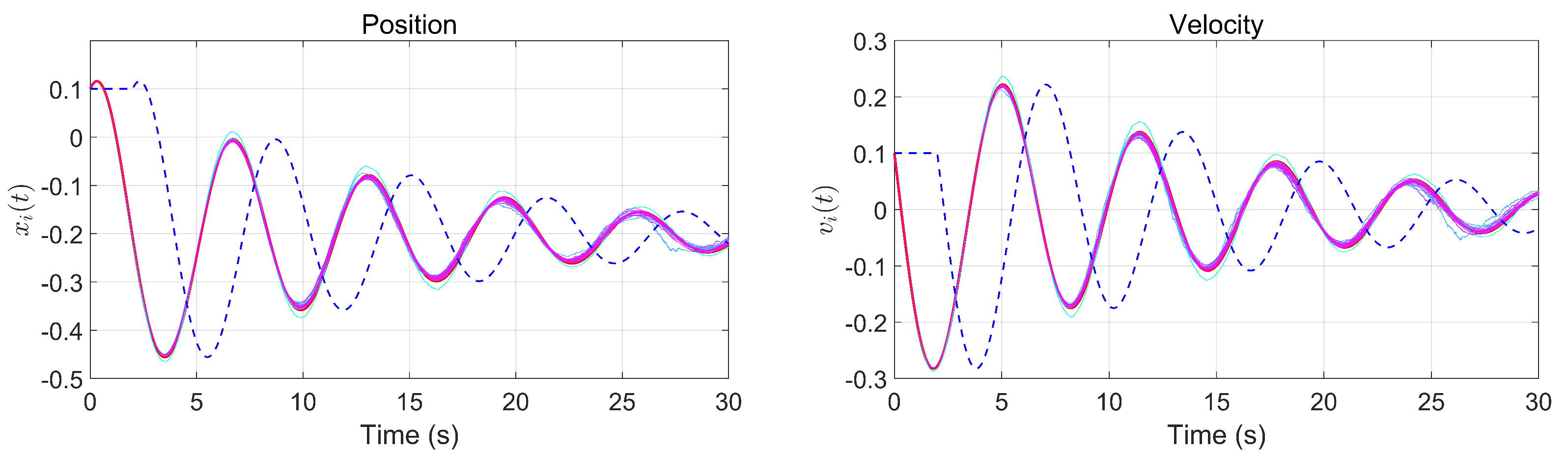}
        \caption{Lag consensus with feedback gain $\kappa_1=5.5$.}
        \label{fig:lag-con-small-word-xv}
    \end{subfigure}
    
    %\vspace{8pt}
    
    \begin{subfigure}[b]{0.93\textwidth}
        \centering
        \includegraphics[width=\linewidth]{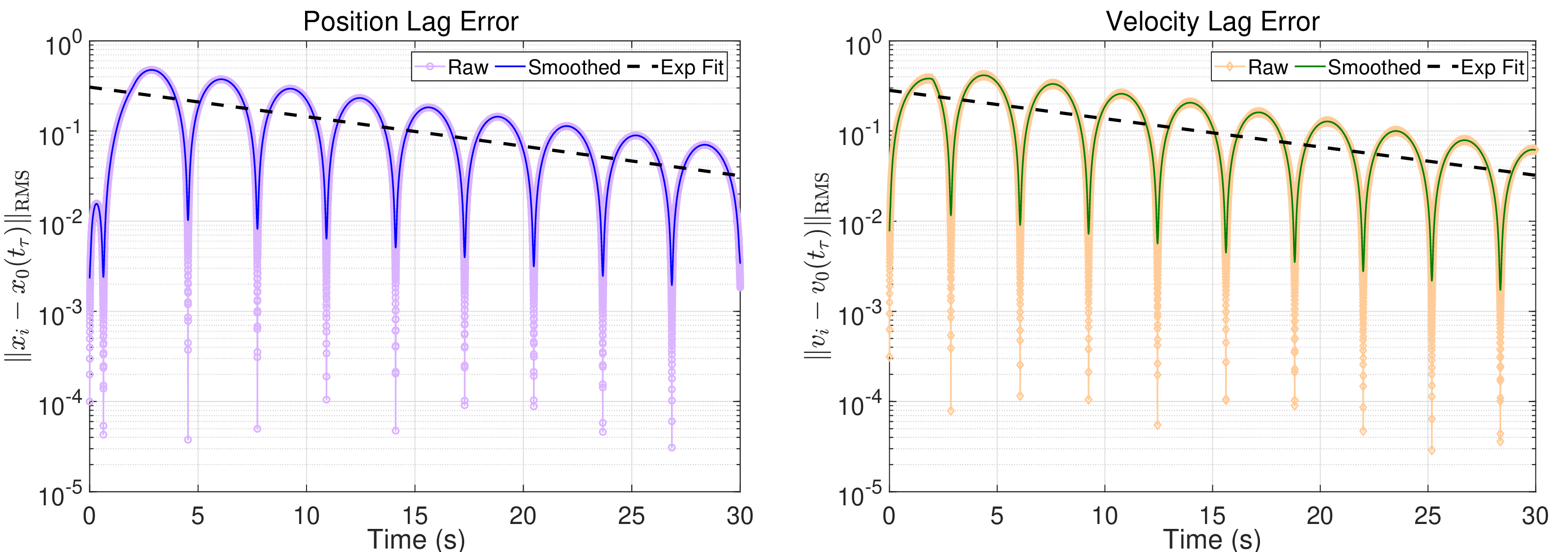}
        \caption{Exponential convergence.}
        \label{fig:lag-con-small-word-error-expon}
    \end{subfigure}
    
    \caption{Mean-square lag consensus in small world based CPS with $50$ followers with lag $\tau=2$ {s}.}
    \label{fig:small-world-50-nodes}
\end{figure*}

\begin{example}[Small world networks]
Without loss of generality, this example utilizes a small-world digraph networks as the agent-based CPS model to explain lag consensus, where the nonlinear dynamics is simulated as 
\begin{align*}
  f(x_i,v_i) = -\sin(x_i) - 0.15v_i - 0.2, \quad i=0,\ldots,50.
\end{align*}
That is, the network consists of agent $ i=0$ serving as the leader and the root node, while the remaining $N=50$ agents are follower agents. We assume that $\kappa_1 >0$, while $\kappa_i = 0$ for $i = 2, \ldots, N$, i.e., only agent $i=1$ receives information from the leader. We generate the intermittent intervals over $t\in[0,30]$ as done in the previous example, and take $\psi=0.7$. We select the parameters $c_1=c_2=1$, $\beta_1=\beta_2=0.9$, $\rho_1= 1$, $\rho_2 = 0.15$ , and evaluate the influence of $\kappa_1$ to fit the suitable conditions in Theorem~\ref{thm:stochastic-lag-consensus-CPS}. The evolutions of second-order lag consensus in the small-world case (using a low feedback gain $\kappa_1=1.7$) is shown in Fig.~\ref{fig:lag-con-small-word-xv-3-1}. It is clear that the followers’ trajectories are perturbed by stochastic noise with large deviations, and they fail to track the leader’s trajectory with a lag $\tau=2$ s, as some followers diverge due to the disturbances. For comparisons, we now increase the feedback gain to $\kappa_1=5.5$. In this scenario, Fig.~\ref{fig:lag-con-small-word-xv} shows that the followers’ trajectories can perfectly track the leader’s trajectory in both position and velocity, while providing robustness against the stochastic perturbations, achieving the mean-square lag consensus. Fig.~\ref{fig:lag-con-small-word-error-expon} reveals the convergence of lag errors in position and velocity; clearly, both decay exponentially in the mean-square sense.
\end{example}

%==========
\section{Conclusion}\label{sec:conclusion}
In this paper, we explored second-order lag consensus in agent-based CPS within the PT-DT framework, where the estimated DT states were stochastically perturbed by the perception layer, resulting in the stochastic lag error dynamics. We thus proposed a resilient lag-consensus protocol that mitigated random noise and intermittent input failures, ensuring mean-square lag consensus under sufficient conditions. Numerical simulations demonstrated the effectiveness of the proposed strategy. 
 Future work will explore the controller’s performance, the effect of delay margins, and its real-world applications, such as connected autonomous vehicles.

%\section*{Acknowledgment}

\appendix
%=================PROOF of THEOREM=========
\subsection{Appendix~A: Proof of Theorem~\ref{thm:stochastic-lag-consensus-CPS}}\label{appendix-proof-thm}
\begin{proof}
\rm Construct a Lyapunov function as
\begin{align}\label{eq:Lyapunov-fun-V}
V(\xi) = \frac{1}{2}\xi(t)^{\top}(P \otimes {I_n})\xi(t),
\end{align}
where %$V(\xi):\mathbb{R}^{2Nn}\rightarrow \mathbb{R}^{+}$, $\xi\in\mathbb{R}^{2n}$, and
\begin{align}
P=\begin{pmatrix}
  2c_{1}c_2\tilde{L}^{s} & c_{1}I_N\\
  c_{1}I_N & c_{2}I_N
\end{pmatrix}.
\label{eq:Lyapunov-P-matrix}
\end{align}
Therefore, we have the following inequalities 
\begin{align}
    \frac{1}{2}\lambda_{\min}(P_1)\|\xi\|^2
    \leq V(\xi)\leq
    \frac{1}{2}\lambda_{\max}(P_2)\|\xi\|^2,
    \label{eq:V_fun_lower-upper-bound}
\end{align}
where the upper and lower bounded matrices are
\begin{align*}
P_1=\begin{pmatrix}
\begin{smallmatrix}
  2c_{1}c_2\lambda_{\min}(\tilde{L}^{s}) & c_{1}I_N\\
  c_{1}I_N & c_{2}I_N
  \end{smallmatrix}
\end{pmatrix},~ 
P_2=\begin{pmatrix}
\begin{smallmatrix}
  2c_{1}c_2\lambda_{\max}(\tilde{L}^{s}) & c_{1}I_N\\
  c_{1}I_N & c_{2}I_N
    \end{smallmatrix}
\end{pmatrix}.
\end{align*}

For the sake of convenience, we define $V_{\xi}(t)=V(\xi(t))$. 
Then, for $t \in [t_{k},s_{k}]$, with the help of
It\^{o} formula based Lyapunov function in \cite{Oksendal2003}, we have
\begin{align}
\der{V}_\xi=\mathcal{L}V\der{t}+\,\xi(t)^{\top}(PU \otimes {I_n})\xi(t)\der{W(t)}.
\end{align}

%------------R_1, R_2========
\begin{figure*}[h!]
    \begin{align}
        R_1 &\!=\! 
        \begin{pmatrix}
     %   \begin{smallmatrix}
{c_1^2{{\tilde L}^s} - {\Delta _1}{I_N}}&{{0_N}}\\
{{0_N}}&{c_2^2{{\tilde L}^s} - ({c_1} +{\Delta _2}){I_N}}
%\end{smallmatrix}
\end{pmatrix},~
R_2\!=\!\begin{pmatrix}
%\begin{smallmatrix}
c_1{\rho_1}{I_N} - c_1^2{L^s} & 2{c_1}{c_2}{K^s} + {\Delta_3}{I_N}\\
\star&({c_1} + {c_2}{\rho _2}){I_N} - c_2^2{L^s}
%\end{smallmatrix}
\end{pmatrix},~
% Q\!=\!%-\frac{1}{2}({H_{1}^T}P+PH_1)\!=\!
% \begin{pmatrix}\begin{smallmatrix}
% {c_1^2{{\tilde L}^s}}&{{0_N}}\\
% {{0_N}}&{c_2^2{{\tilde L}^s} - {c_1}{I_N}}
% \end{smallmatrix}
% \end{pmatrix},~
S\!=\!\begin{pmatrix}
    %\begin{smallmatrix}
        c_1\rho_1{I}_N & \Delta_3{I}_N\\
        \star & c_2\rho_2{I}_N
  %  \end{smallmatrix}
\end{pmatrix}\notag\\
 \Delta_{1}&=c_{1}\rho_{1}+{\Delta_3}, \quad  \Delta_{2}={c_{2} \rho_{2}}+{\Delta_3}, \quad {{\Delta_3}={({{c_1}{\rho _2}}+{{c_2}{\rho _1}})}/{2}}
   \label{eqs:parameters-matrices}
    \end{align}
\end{figure*}

We compute the It\^{o} generator $\mathcal{L}V$ along the trajectories of system \eqref{eq:stochastic-error-dynamics}:
\begin{align}
\mathcal{L}V
&= \xi(t)^{\top} (P \otimes I_n) \big[\tilde{F}_{\tau}(t) + (H_1 \otimes I_n)\xi(t)\big] \notag\\
&\quad + \frac{1}{2} \mathrm{tr} \big[\xi(t)^{\top} (U^{\top} P U \otimes I_n) \xi(t) \big] \notag\\
&\le(c_1\tilde{x}(t)^{\top} +c_2 \tilde{v}(t)^{\top}) 
\big[F(x(t),v(t)) -\bm{f}_N(x_0^{\tau},v_0^{\tau})\big] \notag\\
&\quad - \xi(t)^{\top} (Q \otimes I_n) \xi(t) 
+ \frac{1}{2} \|U^{\top}PU\|\|\xi(t)\|^2.
\label{360}
\end{align}
Here $Q\!=\!-(PH_1)^{s}=\frac{-(H_1^{\top}P+PH_1)}{2}$ or detailed is given by
\begin{align*}
 Q   =\begin{pmatrix}
{c_1^2{{\tilde L}^s}}&{{0_N}}\\
{{0_N}}&{c_2^2{{\tilde L}^s} - {c_1}{I_N}}
%\end{smallmatrix}
\end{pmatrix}.
\end{align*}

% \[Q=-\frac{(H_1^{\top}P+PH_1)}{2},
% \mathrm{tr} \big[\xi(t)^{\top} (U^{\top} P U \otimes I_n) \xi(t)\big] 
% \le \|U^{\top} P U\| \, \|\xi(t)\|^2.
% \]

%%==========================
By Young's inequality, we obtain the following results
\begin{align}
&\quad {c}_1 \, \tilde{x}(t)^\top \Big[F(x(t),v(t)) -\bm{f}_N(x_0^{\tau},v_0^{\tau}))\Big]\notag\\
&\le c_1 \sum_{i=1}^{N} \Big( \big(\rho_1 + \frac{\rho_2}{2}\big) \|\tilde{x}_i(t)\|^2 + \frac{\rho_2}{2} \|\tilde{v}_i(t)\|^2 \Big), \label{ineq:c1_bound} \\[2mm]
&\quad{c}_2 \, \tilde{v}(t)^\top \Big[F(x(t),v(t)) -\bm{f}_N(x_0^{\tau},v_0^{\tau})\Big]\notag\\
&\le c_2 \sum_{i=1}^{N} \Big( \frac{\rho_1}{2} \|\tilde{x}_i(t)\|^2 + \big(\frac{\rho_1}{2} + \rho_2\big) \|\tilde{v}_i(t)\|^2 \Big). \label{ineq:c2_bound}
\end{align}

\par By Assumption~\ref{assum:Lipschitz-cond-nonlinear}, applying inequalities \eqref{ineq:c1_bound} and \eqref{ineq:c2_bound} in \eqref{360}, we obtain
\begin{align}
\mathcal{L}V &\le ({c_1}{\rho_1}+{\Delta_3}) \sum_{i=1}^N \|\tilde{x}_i(t)\|^2
+ ({\Delta_3}+{c_2}{\rho_2}) \sum_{i=1}^N \|\tilde{v}_i(t)\|^2 \notag\\
&\quad - \xi(t)^\top (Q \otimes I_n) \xi(t) + \frac{1}{2} \xi(t)^\top (U^\top P U \otimes I_n) \xi(t) \notag\\
&\le -\xi(t)^\top \Big( R_1 \otimes I_n - \frac{1}{2} U^\top P U \otimes I_n \Big) \xi(t) \notag\\
&\le -\Big(\lambda_{\min}(R_1) - \frac{1}{2} \lambda_{\max}(U^\top P U)\Big) \|\xi(t)\|^2 \notag\\
&= -\frac{\gamma}{2} \|\xi(t)\|^2 \overset{\eqref{eq:V_fun_lower-upper-bound}}{\le} - \frac{\gamma}{\lambda_{\max}(P_2)} V_\xi(t),
\label{ineq:ELV_R_1}
\end{align}
where 
$\gamma := 2 \lambda_{\min}(R_1) - \lambda_{\max}(U^\top P U)$
and the last inequality follows from the Rayleigh quotient bounds 
$\xi^\top R_1 \xi \ge \lambda_{\min}(R_1) \|\xi\|^2$ and $\xi^\top (U^\top P U) \xi \le \lambda_{\max}(U^\top P U) \|\xi\|^2$.  
The matrix $R_1$ is defined in \eqref{eqs:parameters-matrices}.

By inequality \eqref{eq:V_fun_lower-upper-bound} and the monotonicity of expectation, expecting both sides of inequality \eqref{ineq:ELV_R_1}, we have
\begin{align*}
\mathbb{E}\mathcal{L}V
\leq- \frac{\gamma}{2}\mathbb{E}\|\xi(t)\|^2
\le  - \frac{\gamma}{\lambda_{\max}(P_2)}\mathbb{E}V_{\xi}(t),
\end{align*}
and owing to $\mathbb{E}[\xi(t)^{\top}(PU \otimes {I_n})\xi(t)]=0$,
therefore, we have 
\begin{align}
\frac{\der\mathbb{E}V_{\xi}(t)}{\der{t}} &\le -\frac{\gamma}{\lambda_{\max}(P_2)}\mathbb{E}V_{\xi}(t)\triangleq - \bar{\mu}_{1}\mathbb{E}V_{\xi}(t).\label{362}
\end{align}
%Therefore, owing to the inequality $\frac{{d\mathbb{E}V(t)}}{{dt}} \le - {{\bar \mu }_1}\mathbb{E}V(t)$.
\par We now further emphasize on time width with control signals $[t_{k},t_{k+1}]:=I_{1}\cup I_{2}$, i.e.,
% %
% $I_1=\{t \in [t_{k},s_{k}]:\mathbb{E}V_{\xi}(t)=0\}$ and $I_2=\{t \in (s_{k},t_{k+1}):\mathbb{E}V_{\xi}(t)>0\}$. 
\begin{align*}
    I_1&=\{t \in [t_{k},s_{k}]:\mathbb{E}V_{\xi}(t)=0\}, \\
    I_2&=\{t \in (s_{k},t_{k+1}):\mathbb{E}V_{\xi}(t)>0\}.
\end{align*}
If $t\in I_1$, then $\mathbb{E}V_{\xi}(t)=0$.
 Therefore,
\begin{align}
0=\mathbb{E}V_{\xi}(t) \leq \mathbb{E}V_{\xi}(t_k)e^{- \bar{\mu}_1(t-t_{k})}.
\end{align}
\par Otherwise, if $t\in I_2$, by integrating the both sides of the inequality \eqref{362}, we obtain
\begin{align}
\int_{t_k}^{s_k} \frac{\der\mathbb{E}V_{\xi}(t)}{\mathbb{E}V_{\xi}(t)} \leq\int_{t_k}^{s_k} -\bar{\mu}_{1}dt,\notag
\end{align}
and hence
\begin{align}
\mathbb{E}V_{\xi}(t) \le \mathbb{E}V_{\xi}(t_k)e^{- \bar{\mu}_1(t-t_{k})}.\label{368}
\end{align}

\par Similarly, when $t\in(s_k,t_{k+1}), k\in \mathbb{N}$, with the help of the inequalities \eqref{ineq:c1_bound}, \eqref{ineq:c2_bound} and \eqref{eq:V_fun_lower-upper-bound}, the same process gives
\begin{align}
\mathcal{L}V &= \xi(t)^{\top} (P \otimes I_n) \big[\tilde{F}_{\tau}(t) + (H_2 \otimes I_n)\xi(t)\big] \notag \\
&= (c_1 \tilde{x}(t)^{\top} + c_2 \tilde{v}(t)^{\top}) \big[F(x(t),v(t)) -\bm{f}_N(x_0^{\tau},v_0^{\tau})\big]  \notag \\
&\quad + \xi(t)^{\top} \big((P H_2)^s \otimes I_n\big) \xi(t) \notag \\
&\le \xi(t)^{\top} \Big(\big(S +(PH_2)^s\big)\otimes I_n\Big) \xi(t) \notag \\
%&\leq \xi(t)^{\top} (R_2 \otimes I_n) \xi(t) \notag \\
&\leq \lambda_{\max}(R_2) \|\xi(t)\|^2\overset{\eqref{eq:V_fun_lower-upper-bound}}\le \frac{2\lambda_{\max}(R_2)}{\lambda_{\min}(P_1)} V_{\xi}(t), \label{ineq:ELV-R_2}
\end{align}
where the details of $R_2=S+(PH_2)^{s}$ and $S$  are shown in \eqref{eqs:parameters-matrices}.
Taking expectation value for both sides of \eqref{ineq:ELV-R_2} yields
\begin{align}
\frac{\der\mathbb{E}V_{\xi}}{\der{t}} \leq \frac{2\lambda_{\max}(R_2)}{\lambda_{\min}(P_1)}\mathbb{E}V_{\xi}(t) \triangleq\bar{\mu}_2\mathbb{E}V_{\xi}(t),\label{364}
\end{align}
which implies that
\begin{align}
\mathbb{E}V_{\xi}(t) \le \mathbb{E}V_{\xi}(s_{k})e^{- \bar{\mu}_2(t - s_k)}.\label{369}
\end{align}
\par From \eqref{368} and \eqref{369}, for $t \in [t_k,t_{k+1}), k\in \mathbb{N}$, we have
\begin{align*}
\mathbb{E}V_{\xi}(t)&\le{\max}\left\{\mathbb{E}V_{\xi}(t_{k})e^{- \bar{\mu}_1(t - t_{k})},\mathbb{E}V_{\xi}(s_{k})e^{- \bar{\mu}_2(t - s_{k})}\right\}\\
&\le {M_k}e^{-\breve{\mu}t},
\end{align*}
where $M_{k}=\max\big\{\mathbb{E}V_{\xi}(t_{k})e^{\bar{\mu}_{1}t_{k}},\mathbb{E}V_{\xi}(s_k)e^{\bar{\mu}_{2}s_{k}}\big\}$, $k\in \mathbb{N}$ and $\breve{\mu}\triangleq{\min}\left\{\bar{\mu}_1,\bar{\mu}_2\right\}$.
Also, we define ${\bar M}={\sup _k}{M_k}$, and show ${\bar M}<\infty$. In fact,
combing \eqref{369} with \eqref{368}, we derive that
\begin{align}\label{390}
&\quad~\mathbb{E}V_{\xi}(t_{l+1})\notag\\
&\le \mathbb{E}V_{\xi}(s_l) 
   \exp\big[\bar\mu_2 (t_{l+1} - s_l)\big] \notag\\
&\le \mathbb{E}V_{\xi}(t_l) 
   \exp\big[-\bar\mu_1 (s_l - t_l) + \bar\mu_2 (t_{l+1} - s_l)\big] \notag\\
&\le \mathbb{E}V_{\xi}(t_{l-1})
   \exp\Big\{
      \big[-\bar{\mu}_1 (s_{l-1} - t_{l-1}) + \bar{\mu}_2 (t_l - s_{l-1})\big] \notag\\
&\qquad\qquad \qquad+ \big[-\bar{\mu}_1 (s_l - t_l) + \bar{\mu}_2 (t_{l+1} - s_l)\big] \Big\} \leq \ldots \notag\\
&\le \mathbb{E}V_{\xi}(0)
   \exp\Big\{
      -\bar{\mu}_1 \sum_{k=0}^l (s_k - t_k)
      + \bar{\mu}_2 \sum_{k=0}^l (t_{k+1} - s_k)
   \Big\} \notag\\
&\le \mathbb{E}V_{\xi}(0)
   \exp\big[-(l+1)(\bar{\mu}_1 \theta - \bar{\mu}_2 (\delta - \theta))\big].
\end{align}

\par As a result, we obtain
\begin{align}
\mathbb{E}V_{\xi}(t)\le \bar{M}e^{-\breve{\mu}t},\quad \forall t \ge 0.
\end{align}
\par On the other hand, from \eqref{eq:V_fun_lower-upper-bound} and \eqref{eq:Lyapunov-fun-V}, we have
\begin{align*}
\mathbb{E}V_{\xi}(t)\ge \frac{1}{2}\lambda _{\min}(P_1)\mathbb{E}\|\xi(t)\|{^2},
\end{align*}
and hence
\begin{align}
\mathbb{E}\|{\xi}(t)\|{^2}\leq\frac{2\bar M}{\lambda _{\min}(P_1)}e^{-\breve{\mu}t}, \quad \forall t\ge0.
\end{align}
\par It also leads to exponential stability of \eqref{eq:stochastic-error-dynamics} since
\begin{align*}
 \frac{1}{t}\log (\mathbb{E}\|\xi(t)\|{^2})\leq\frac{1}{t}\bigg(\log \frac{2 \bar{M}}{\lambda_{\min}(P_1)}-{\breve{\mu}}t\bigg)\to{-\breve{\mu}}<0,
\end{align*}
as $t\to \infty$. It further implies that the followers' position $x(t)$ and velocity $v(t)$ globally and exponentially converge in mean square to the leader one $x_{0}(t_\tau)$ and $v_0(t_\tau)$ respectively, as $t\to {+}\infty $. Hence, second-order agent networks \eqref{eq:follower-dynamics} and \eqref{eq:leader-reference} can reach mean square lag consensus. %The proof is completed.
%\hfill$\blacksquare$
\end{proof}

 % \section*{Acknowledgments}
 % This work was initiated while Z. Zhang was a visiting researcher
 % in M. Nagahara's group, supported by a grant from The Study Abroad Program for Graduate Student of GUET

\bibliographystyle{ieeetr}
\bibliography{reference}

\end{document}